\documentclass[letterpaper, 10 pt, conference]{ieeeconf}
\usepackage{amsmath,amsfonts,amssymb}

\usepackage{hyperref}
\usepackage{cite}

\usepackage{algorithm}
\usepackage{algpseudocode}

\usepackage{booktabs}
\usepackage{multirow}
\usepackage{tabularx}

\usepackage{graphicx}
\usepackage[caption=false,font=footnotesize]{subfig}
\usepackage{tikz, pgfplots, tikz-3dplot}
\usetikzlibrary{shapes, arrows, arrows.meta, fit, positioning, calc, 3d}
\usepgfplotslibrary{fillbetween}
\pgfplotsset{compat=1.18} 
\tikzset{>=latex}

\providecommand*{\M}[1]{\mathbf#1}     
\providecommand*{\V}[1]{\boldsymbol#1} 
\providecommand*{\T}[1]{\mathrm{#1}}   

\DeclareMathOperator*{\argmin}{arg\,min}
\newtheorem{theorem}{Theorem}
\newtheorem{lemma}{Lemma}
\newtheorem{definition}{Definition}

\newtheorem{assumption}{Assumption}

\allowdisplaybreaks
\IEEEoverridecommandlockouts
\begin{document}
\bstctlcite{IEEEexample:BSTcontrol}

\title{\LARGE \bf
Computing the Exact Pareto Front in Average-Cost\\ Multi-Objective Markov Decision Processes
}
\author{
\IEEEauthorblockN{Jiping~Luo~and~Nikolaos~Pappas}
\IEEEauthorblockA{Department of Computer and Information Science, Link{\"o}ping University, Sweden\\
Email: jiping.luo@liu.se, nikolaos.pappas@liu.se}
}

\author{Jiping~Luo~and~Nikolaos~Pappas
\thanks{This work has been supported by ELLIIT, CUGS, and 6G-LEADER.}
\thanks{J. Luo and N. Pappas are with the Department of Computer and Information Science, Link\"oping University, Link\"oping 58183, Sweden. Emails: {\tt\small jiping.luo@liu.se} and {\tt \small nikolaos.pappas@liu.se}.}%
}

\maketitle
\begin{abstract}
    Many communication and control problems are cast as multi-objective Markov decision processes (MOMDPs). The complete solution to an MOMDP is the Pareto front. Much of the literature approximates this front via scalarization into single-objective MDPs. Recent work has begun to characterize the full front in discounted or simple bi-objective settings by exploiting its geometry. In this work, we characterize the exact front in average-cost MOMDPs. We show that the front is a continuous, piecewise-linear surface lying on the boundary of a convex polytope. Each vertex corresponds to a deterministic policy, and adjacent vertices differ in exactly one state. Each edge is realized as a convex combination of the policies at its endpoints, with the mixing coefficient given in closed form. We apply these results to a remote state estimation problem, where each vertex on the front corresponds to a threshold policy. The exact Pareto front and solutions to certain non-convex MDPs can be obtained without explicitly solving any MDP.
\end{abstract}

\vspace{-0.02in}
\section{Introduction}
\vspace{-0.04in}
Achieving a desired trade-off among multiple conflicting objectives is fundamental in many communication and control systems~\cite{baillieul2007control, park2018wireless, kountouris2021semantics, luo2025semantic}. In practice, there is rarely a free lunch: performance improvements are often achieved at the expense of increased resource consumption. Many such problems are cast as Markov decision processes (MDPs). 

An MDP is a controlled Markov chain with state process $\{X_t\}_{t\geq 1}$ evolving on a state space ${\sf X}$, action process $\{U_t\}_{t\geq 1}$ taking values in an action set ${\sf U}$, and one-stage cost $c(X_t,U_t)$ incurred at time $t$. At each time $t$, an action $U_t$ is taken according to a decision rule $U_t \sim \pi_t(\cdot|I_t)$, where $I_t = \{X_{1:t}, U_{1:t-1}\}$ denotes the history prior to decision epoch $t$. The $\gamma$-discounted optimal control problem seeks a \emph{policy} $\pi = \{\pi_t\}_{t\geq 1}$ that minimizes the objective function
\begin{equation}
    J_{\beta}^\gamma(\pi) := \sum_{t=1}^\infty\gamma^{t-1} {\sf E}^\pi \big[c(X_t, U_t)\mid X_1 \sim \beta\big],\label{eq:discounted-cost}
\end{equation}
where $\gamma \in (0, 1)$ is the discount factor, and $\beta$ is the initial state distribution. In many applications~\cite{schenato2008optimal, leong2017sensor, luo2025remote, luo2025cost}, the average-cost criterion is more appropriate
\begin{equation}
    J_\beta(\pi) := \limsup_{T\to \infty}\frac{1}{T}\sum_{t=1}^T {\sf E}^\pi \big[c(X_t, U_t)\mid X_1 \sim \beta \big].\label{eq:average-cost}
\end{equation}

A multi-objective MDP (MOMDP) considers multiple costs. Let $\V{c}(X_t, U_t) = (c_1(X_t, U_t), \ldots, c_K(X_t, U_t))$ denote the one-stage cost vector, and let $J_{\beta, k}(\pi)$ denote the $k$th objective. The goal is to minimize the objective vector
\begin{equation}
\V{J}_{\beta}(\pi) = \big(J_{\beta, 1}(\pi), \ldots, J_{\beta, K}(\pi)\big).\label{eq:vector-cost}
\end{equation}
The main conceptual difficulty is that, because the objectives may be conflicting, there typically does not exist a single policy that minimizes all components of $\V{J}_\beta(\pi)$. Instead, optimality is defined in the Pareto sense. A policy $\pi^*$ is Pareto optimal if there does not exist any other policy $\pi$ such that
\begin{equation}
J_{\beta, k}({\pi}) \leq J_{\beta, k}(\pi^*)\,\,\,\text{for all}\,\,\,k = 1,\ldots, K, 
\end{equation}
and $J_{\beta, i}({\pi}) < J_{\beta, i}(\pi^*)$ for at least one index $i$. The complete solution to an MOMDP is the set of all Pareto optimal policies, whose objective vectors form the \emph{Pareto front}.

Much of the literature operates under the view that computing the \emph{exact} Pareto front is computationally intractable except for very small problems~\cite{white1982multi, chatterjee2006markov, wiering2007computing, roijers2013survey, van2014multi}. Most studies \emph{approximate} the front using \emph{linear scalarization} methods~\cite{roijers2013survey}. Let $\V{w} = (w_1, \ldots, w_K)$ denote a weight vector representing the user's preference over the objectives. This method seeks to minimize
\begin{equation}
    J_{\beta}(\pi; \V{w}) = \langle \V{w}, \V{J}_\beta(\pi) \rangle,
\end{equation}
where $\langle \V{a}, \V{b} \rangle = \sum_i a_i b_i$ denotes the inner product. Clearly, the scalarized problem is a single-objective MDP with additive cost $c^{\V{w}}(x, u) = \langle \V{w}, \V{c}(x, u)\rangle$. This makes the approach particularly attractive, as it allows standard dynamic programming (DP) or reinforcement learning (RL) techniques to generate some representative Pareto-optimal policies.

However, this method has two main limitations. First, traversing all weight vectors is neither tractable nor sufficient to construct the entire Pareto front~\cite{luo2025value}. Second, many applications exhibit nonlinear relationships among the objectives~\cite{roijers2013survey}. For example, increasing communication frequency accelerates hardware aging superlinearly~\cite{luo2025remote}, resulting in additional maintenance costs or reduced equipment longevity that cannot be adequately captured by linear weights. Let $f: \mathbb{R}^{K} \to \mathbb{R}$ denote a nonlinear scalarization function. The nonlinear scalarized problem seeks to minimize 
\begin{equation}
    J_{\beta}(\pi; f) = f(\V{J}_{\beta}(\pi)). \label{eq:nonlinear scalarized obj}
\end{equation}
Classical DP methods cannot handle such nonlinear scalarized MDPs because nonlinear scalarization breaks the additivity required by the Bellman equation~\cite{roijers2013survey}.

This work studies the exact Pareto front in MOMDPs. The recent papers~\cite{luo2025value} and~\cite{li2025how} investigate related problems in different settings; \cite{luo2025value} analyzes an average-cost bi-objective MDP via DP and Pareto analysis, while \cite{li2025how} studies a $\gamma$-discounted MOMDP using linear scalarization and DP. In contrast, we focus on average-cost MOMDPs using occupancy measures and linear programming (LP). Our main contributions are as follows.

(1) We characterize the geometry of the Pareto front in average-cost MOMDPs. Under standard assumptions, the set of achievable objective vectors forms a convex polytope, and the Pareto front is a continuous, piecewise linear surface lying on the boundary of this polytope. Each vertex of the polytope corresponds to a deterministic policy, and policies at adjacent vertices differ in exactly one state. Each edge is realized by a convex combination of the policies at its endpoints, and the mixing coefficient is given in closed form.

(2) We show that for nonlinear scalarized MDPs, if the scalarization function $f$ is strictly increasing, the optimal solution lies on the Pareto front and can be realized as a convex combination of at most $K$ deterministic policies.

(3) We apply these results to a remote state estimation problem, in which each vertex on the Pareto front corresponds to a threshold policy. We show that the exact Pareto front and the optimal solution to the nonlinear scalarized MDP can be obtained without explicitly solving any MDP.

\textit{Organization:} Section~\ref{sec:preliminary} presents preliminaries on MDPs and occupancy measures. Our main results are summarized in Section~\ref{sec:main results}, and a remote estimation example is examined in Section~\ref{sec:example}. Proofs are provided in the Appendix.

\section{Preliminaries}\label{sec:preliminary}
\subsection{MDP}
An average-cost MOMDP is a tuple $({\sf X}, {\sf U}, \M{P}, \V{c}, \beta)$, where ${\sf X}$ and ${\sf U}$ denote the state and action spaces, $\beta$ is the initial state distribution, and $\V{c}(x, u) = (c_1(x, u), \ldots, c_K(x, u))$ is the cost vector. The transition kernel $\M{P}$ governs the evolution of the state process, where $\M{P}(x, u, x^\prime)$ is the probability of transitioning from state $x$ to $x^\prime$ given action $u$. The objective vector under a policy $\pi$, $\V{J}_\beta(\pi)$, is defined in~\eqref{eq:average-cost}--\eqref{eq:vector-cost}. A linear scalarized MDP with weight vector $\V{w}$ is described by the tuple $({\sf X}, {\sf U}, \M{P}, c^{\V{w}}, \beta)$, where $c^{\V{w}}(x, u) = \langle \V{w}, \V{c}(x, u) \rangle$.

Next, we introduce several policy classes that will be used throughout the paper. Let $\Pi$ denote the set of all admissible policies, which may be history-dependent and non-stationary. We define the following subclasses:
\begin{itemize}
    \item $\Pi_{\T M}$: Markov policies; $\pi \in \Pi_{\T M}$ if it depends on the history $I_t$ only through $X_t$, i.e., $U_t \sim \pi_t(\cdot|X_t)$.
    \item $\Pi_{\T S}$: Markov \emph{stationary} policies; $\pi \in \Pi_{\T S}$ if it uses the same rule at all times, i.e., $U_t \sim \pi(\cdot|X_t)$.
    \item $\Pi_{\T D}$: Markov stationary \emph{deterministic} policies; $\pi \in \Pi_{\T D}$ if it selects an action with certainty, i.e., $U_t = \pi(X_t)$.
\end{itemize}

Clearly, $\Pi_{\rm D} \subset \Pi_{\rm S} \subset \Pi_{\rm M} \subset \Pi$. For clarity, we refer to policies in $\Pi_{\rm S}$ as stationary and those in $\Pi_{\rm D}$ as deterministic. Moreover, two deterministic policies $\pi_1, \pi_2 \in \Pi_{\T D}$ are said to be \emph{adjacent} if they differ in exactly one state; that is, there exists a state $i_0 \in {\sf X}$ such that $\pi_1(x) = \pi_2(x)$ for all $x \neq i_0$ and $\pi_1(i_0) \neq \pi_2(i_0)$.

For any $\bar{\Pi} \subseteq \Pi$, let 
\begin{equation}
    \mathcal{J}_\beta^{\bar{\Pi}} := \{\V{J}_\beta(\pi):\pi \in \bar{\Pi}\}
\end{equation}
denote the set of objective vectors achievable by $\bar{\Pi}$. 

\begin{definition}
    $\M{P}$ is \emph{unichain} if for every $\pi \in \Pi_{\T S}$, $\{X_t\}$ governed by $\M{P}_\pi$, where $\M{P}_\pi(x, x^\prime) = \sum_{u}\pi(u|x)\M{P}(x, u, x^\prime)$, forms a Markov chain with a single recurrent class.
\end{definition}

Unless otherwise stated, it is assumed in this paper that
\begin{assumption} 
    ${\sf X}$ and ${\sf U}$ are finite sets and $\M{P}$ is unichain.
\end{assumption}

We now introduce a class of mixing policies that will play an important role in our analysis.
\begin{definition}\label{def:mixing policy}
    Let $\alpha \in (0,1)$ and $\pi_1, \pi_2 \in \Pi_{\T S}$. Then $\pi = \alpha \pi_1 + (1-\alpha)\pi_2 \in \Pi_{\T S}$ is called a \emph{mixing policy}. It is called a \emph{simple mixing policy} if $\pi_1, \pi_2 \in \Pi_{\T D}$ and they are adjacent.
\end{definition}

For any simple mixing policy, randomization occurs only at state $i_0$ (where $\pi_1$ and $\pi_2$ differ). As a result, $\{X_t\}$ forms a regeneration process with regeneration events defined by the successive returns to $i_0$. The sample path can therefore be decomposed into i.i.d. cycles, the expected length of which is determined by the selected policy.

\subsection{Occupancy Measure}
For any policy $\pi$ and initial distribution $\beta$, define the state-action frequency up to time $t$ by
\begin{equation}
    \mu_{\beta, \pi}^t(x, u) := \frac{1}{t}\sum_{n=1}^t {\sf P}^\pi_\beta (X_n = x, U_n = u),
\end{equation}
and the state frequency by $\mu_{\beta, \pi}^t(x):=\sum_{u\in {\sf U}}\mu_{\beta, \pi}^t(x, u)$. The vector $\mu_{\beta,\pi}^t= (\mu_{\beta,\pi}^t(x,u): x\in {\sf X}, u \in {\sf U})$ is called the \emph{finite-horizon occupancy measure}. 

We define the \emph{(limiting) occupancy measure} under $\pi$ and $\beta$ as the set of all subsequential limits of $\{\mu_{\beta,\pi}^t\}$, i.e.,
\begin{equation}
    \Gamma_{\beta}^\pi := \big\{
        \lim_{n \to \infty} \mu^{t_n}_{\beta, \pi}:\{t_n\} \subseteq \mathbb{N}~\text{with}~t_n \to \infty
    \big\}. 
\end{equation}
The limits are defined componentwise. Subsequential limits are used because the ordinary limit $\lim_{t\to\infty}\mu_{\beta,\pi}^t$ need not exist for non-stationary policies. When the limit does exist, we denote it by $\mu_{\beta,\pi}$. 

For any $\bar{\Pi} \subseteq \Pi$, the set of occupancy measures achievable under policies in $\bar{\Pi}$ is denoted by
\begin{equation}
    \Gamma_{\beta}^{\bar{\Pi}} := \bigcup_{\pi \in \bar{\Pi}}\Gamma_{\beta}^\pi.
\end{equation}
For brevity, we denote the full set as $\Gamma_\beta = \Gamma_\beta^\Pi$.

The following lemma characterizes the existence of the ordinary limit (see, e.g.,~\cite[Ch.~3.4]{bartle2011introduction}).

\begin{lemma}\label{lemma:sequential_limits}
    $\Gamma_{\beta}^\pi$ is non-empty for all $\pi \in \Pi$. The largest and smallest subsequential limits are 
    \begin{equation*}
        \sup(\Gamma_{\beta}^\pi) = \limsup_{t\to \infty} \mu^t_{\beta, \pi}\,\,\,\text{and}\,\,\,\inf(\Gamma_{\beta}^\pi) = \liminf_{t\to \infty} \mu^t_{\beta, \pi}.
    \end{equation*}
    Moreover, the ordinary limit $\lim_{t\to\infty}\mu_{\beta,\pi}^t$ exists if and only if $\Gamma_{\beta}^\pi$ is a singleton, i.e., $\sup(\Gamma_{\beta}^\pi) = \inf(\Gamma_{\beta}^\pi) = \mu_{\beta, \pi}$. 
\end{lemma}

For stationary policies, we have the following property.
\begin{lemma}\label{lemma:singleton}
    For any $\pi \in \Pi_{\T S}$, $\Gamma_{\beta}^\pi = \{\mu_\pi\}$ for all $\beta$, where 
    \begin{equation}
        \mu_{\pi}(x, u) = \nu_\pi(x) \pi(u|x),
    \end{equation}
    and $\nu_\pi$ is the stationary distribution of $\{X_t\}$ under $\pi$. 
\end{lemma}

Lemma~\ref{lemma:singleton} implies that under any stationary policy $\pi \in \Pi_{\T S}$, each objective $J_{\beta, k}(\pi)$ is independent of $\beta$ and can be expressed linearly via the occupancy measure $\mu_\pi$ as
\begin{equation}
    J_k(\pi) = \langle \mu_\pi, c_k\rangle := \sum_{x, u} \mu_\pi(x, u) c_k(x, u).
\end{equation}
For a general non-stationary policy $\pi$, the sequence $\{\mu_{\beta, \pi}^t\}$ may not converge, and $\Gamma_\beta^\pi$ is not a singleton. As a result, the average cost satisfies
\begin{equation*}
    J_{\beta, k}(\pi) \geq \langle \mu, c_k\rangle\,\,\,\text{for all}\,\,\,\mu\in \Gamma_{\beta}^\pi.
\end{equation*}

\section{Properties of the Pareto Front}\label{sec:main results}
\subsection{Geometric Results}
This section establishes geometric properties of the set of occupancy measures and the set of objective vectors. We begin by introducing an auxiliary set $\Phi \subset \mathbb{R}^{|{\sf X}||{\sf U}|}$, which underpins our subsequent analysis. The set $\Phi$ is defined as
\begin{equation}
    \Phi \hspace{-0.3em}:=\hspace{-0.3em} \left\{
    \begin{aligned}
    &\phi(x, u): (x, u)\in {\sf X} \times {\sf U},\\
    &\sum_{u}\phi(x^\prime, u) \hspace{-0.25em}=\hspace{-0.25em} \sum_{x, u}\phi(x, u)\M{P}(x, u, x^\prime), x^\prime \in {\sf X},\\
    &\sum_{x, u}\phi(x, u) \hspace{-0.25em}=\hspace{-0.25em} 1, \,\phi(x, u) \geq 0, (x, u)\in {\sf X} \times {\sf U}
    \end{aligned}
    \right\}\label{eq:Phi set}
\end{equation}
We write the marginal over states as $\phi(x) = \sum_{u}\phi(x, u)$.

A useful property of $\Phi$ is the following.
\begin{lemma}
    $\Phi$ is a convex polytope.
\end{lemma}

The following standard definitions are needed.
\begin{definition}
    Let $\mathcal{P} \subset \mathbb{R}^{D}$ denote a $D$-dimensional polytope defined by the intersection of a finite number of linear constraints. A point $p \in \mathcal{P}$ is a \emph{vertex} (or extreme point) of $\mathcal{P}$ if it cannot be expressed as a convex combination of two other points in $\mathcal{P}$; that is, there do not exist $p_1, p_2\in \mathcal{P}$, $p_1 \neq p_2$, and $\alpha \in (0, 1)$ such that $p = \alpha p_1 + (1-\alpha) p_2$. Let $\mathcal{V}(\mathcal{P})$ denote the set of vertices of $\mathcal{P}$. Each vertex corresponds to at least $D$ active constraints, among which $D$ are linearly independent. Two vertices are \emph{adjacent} if there exists a set of $D-1$ linearly independent constraints that are active at both vertices. For any $v \in \mathcal{V}(\mathcal{P})$, let $\mathcal{N}(v) \subset \mathcal{V}(\mathcal{P})$ denote the set of vertices adjacent to $v$. For any $v^\prime \in \mathcal{N}(v)$, the line segment joining $v$ and $v^\prime$ is an \emph{edge} of $\mathcal{P}$, denoted by $\mathcal{G}(v, v^\prime)$. A subset $\mathcal{F} \subset \mathcal{P}$ is a \emph{face} of $\mathcal{P}$ if there exists a $c \in \mathbb{R}^D$ such that $\mathcal{F} = \argmin_{p \in \mathcal{P}} \langle c, p \rangle$. Any proper face of $\mathcal{P}$ has dimension $0 \leq d \leq D - 1$. A $0$-dimensional, $1$-dimensional, and $(D-1)$-dimensional face is a vertex, an edge, and a \emph{facet}, respectively. The \emph{boundary} of $\mathcal{P}$ is the union of all its proper faces. The operation $\operatorname{conv}(\mathcal{S})$ denotes the \emph{convex hull} of a set of points $\mathcal{S} \subseteq \mathcal{P}$.
\end{definition}

A main result of this section follows. Theorem~\ref{theorem:occupancy measure}(a) establishes the equivalence between the set of occupancy measures $\Gamma_\beta$ and the auxiliary set $\Phi$. Theorem~\ref{theorem:occupancy measure}(b)--(d) further elucidate the relationships between deterministic policies and the vertices, adjacent vertices, and edges of $\Phi$.

\begin{theorem}\label{theorem:occupancy measure}
    $\Gamma_\beta$ and $\Phi$ have the following properties:
    \begin{enumerate}
        \item[(a)] For all $\beta$, $\Gamma_\beta = \Gamma_\beta^{\Pi_{\T S}} = \operatorname{conv}(\Gamma_\beta^{\Pi_\T{D}})= \Phi$. Each $\mu \in \Gamma_\beta$ is realizable by a stationary policy $\pi_\mu \in \Pi_{\T S}$, where
        \begin{equation}
            \pi_{\mu}(u|x) = \begin{cases}
                \mu(x, u) / \mu(x), &\mu(x)>0,\\
                {\rm arbitrary}, &\mu(x)=0.
            \end{cases} \label{eq:recover policy}
        \end{equation}
        \item[(b)] $\mathcal{V}(\Phi) = \Gamma_\beta^{\Pi_{\T{D}}}$ for any $\beta$. 
        \item[(c)] Any two adjacent vertices of $\Phi$ differ in one state.
        \item[(d)] Let $\phi_1, \phi_2 \in \mathcal{V}(\Phi)$ be adjacent vertices differing in state $i_0$, and let $\pi_{\phi_1}, \pi_{\phi_2} \in \Pi_{\T D}$ denote the corresponding deterministic policies. Each $\phi \in \mathcal{G}(\phi_1, \phi_2)$, i.e., $\phi = b \phi_1 + (1-b) \phi_2$ for $b \in (0, 1)$, can be realized by a simple mixing policy $\pi = \alpha \pi_{\phi_1} + (1-\alpha) \pi_{\phi_2}$, where
        \begin{equation}
            \alpha = \frac{b \nu_{\pi_{\phi_1}}(i_0)}{b \nu_{\pi_{\phi_1}}(i_0) + (1-b) \nu_{\pi_{\phi_2}}(i_0)}. \label{eq:mixing coefficient}
        \end{equation}
    \end{enumerate}
\end{theorem}

Note that Theorem~\ref{theorem:occupancy measure}(a) does not imply a one-to-one correspondence between the set of policies $\Pi$ and the set of occupancy measures $\Phi$. While each stationary policy $\pi \in \Pi_{\T S}$ produces a unique occupancy measure $\mu_\pi \in \Phi$, some history-dependent non-stationary policies may yield the same occupancy measure. Conversely, every occupancy measure $\mu \in \Phi$ can be realized by some stationary policy $\pi_\mu \in \Pi_{\T S}$ defined by~\eqref{eq:recover policy}, though the choice is not unique.

Theorem~\ref{theorem:occupancy measure}(a) implies that there is no loss of optimality in restricting attention to stationary policies. Consequently, solving the (MO)MDP over the full policy space $\Pi$ is equivalent to solving an (MO)LP over the convex polytope $\Phi$. Formally, a linear scalarized MDP is equivalent to the LP
\begin{equation}
    \min_{\phi \in \Phi}\,\, J(\phi; \V{w}) = \langle \phi, c^{\V{w}} \rangle.
\end{equation}
An MOMDP can be expressed as the MOLP
\begin{equation}
    \min_{\phi \in \Phi}\,\,\V{J}(\phi) = \big( J_1(\phi), \ldots, J_K(\phi)\big), 
\end{equation}
where $J_k(\phi) = \langle \phi, c_k \rangle$. We will write $J(\phi; \V{w}) = \langle \V{w}, \V{J}(\phi) \rangle$.

\begin{lemma}
    $\V{J}: \Phi \to \mathbb{R}^K$ is a linear mapping.
\end{lemma}

Let $\mathcal{J}^{\Phi} = \{\mathbf{J}(\phi): \phi \in \Phi\}$ denote the set of achievable objective vectors. We denote by $\mathcal{V}^*(\Phi) \subseteq \mathcal{V}(\Phi)$ the set of Pareto optimal vertices, and by $\mathcal{N}^*(\phi) = \mathcal{N}(\phi) \cap \mathcal{V}^*(\Phi)$ the set of Pareto optimal vertices adjacent to vertex $\phi \in \mathcal{V}(\Phi)$. 

The following result summarizes the key structural properties of $\mathcal{J}^{\Phi}$ and the Pareto front. Theorem~\ref{theorem:frontier}(a) implies that any vertex of $\mathcal{J}^\Phi$ is induced by a vertex of $\Phi$, but the converse may not hold. Theorem~\ref{theorem:frontier}(c)--(f) establish that the Pareto front is a \emph{continuous, piecewise linear} surface that lies on the boundary of the convex polytope $\mathcal{J}^\Phi$. 

\begin{theorem}\label{theorem:frontier}
The Pareto front has the following properties:
\begin{enumerate}
    \item[(a)] $\mathcal{J}^{\Phi}$ is a convex polytope in $\mathbb{R}^K$ and $\mathcal{V}(\mathcal{J}^\Phi) \subseteq \mathcal{J}^{\mathcal{V}(\Phi)}$.
    \item[(b)] For all $\beta$, $\mathcal{J}_\beta^\Pi = \mathcal{J}_\beta^{\Pi_{\T S}}= \operatorname{conv}(\mathcal{J}_\beta^{\Pi_{\T D}}) = \mathcal{J}^\Phi$.
    \item[(c)] The Pareto front lies on the boundary of $\mathcal{J}^\Phi$.
    \item[(d)] Each $\phi \in \mathcal{V}^*(\Phi)$ minimizes $J(\phi; \V{w})$ for some $\V{w}\succ \V{0}$.
    \item[(e)] Each $\phi \in \mathcal{V}^*(\Phi)$ has at least one adjacent Pareto optimal vertex, i.e., $1 \leq |\mathcal{N}^*(\phi)| \leq |{\sf X}|(|\sf U| - 1)$.
    \item[(f)] Let $\phi, \phi^\prime \in \mathcal{V}^*(\Phi)$ be adjacent Pareto optimal vertices. Then the edge $\mathcal{G}(\phi, \phi^\prime)$ lies on the Pareto front.
\end{enumerate}
\end{theorem}

The vertices and edges constitute the 1-skeleton of the Pareto front. Crucially, policies on this skeleton are \emph{quasi-deterministic}: any point on an edge can be realized by a simple mixing policy that randomizes in at most a single state, and the mixing coefficient is given in closed form. Although the full front consists of higher-dimensional faces, the skeleton provides a sufficient representative solution set because (i) the associated policies are easy to compute and implement; (ii) any optimal solution to a linear scalarized MDP is attained at a vertex; and (iii) every point on the Pareto front is a convex combination of these vertices.

We show in the next subsection, however, that an optimal solution for a general nonlinear scalarized MDP may reside within the interior of a face. In such cases, the optimal policy may require randomizing in up to $K$ states, necessitating a more comprehensive search over the full Pareto front.

\subsection{Scalarized MDPs}
In general, the scalarization function $f$ is assumed to be strictly increasing (but not necessarily convex). The intuition is that the user's preference improves if the cost of one objective decreases while the other objective values remain unchanged.

\begin{definition}
    The scalarization function $f$ is said to be \emph{strictly increasing} if for any objective vectors $\V{J}, \V{J}^\prime \in \mathcal{J}^\Phi$ such that $\V{J} \preceq \V{J}^\prime$ and $\V{J} \neq \V{J}^\prime$, it holds that $f(\V{J}) < f(\V{J}^\prime)$.
\end{definition}


The next theorem shows that we can restrict attention to the Pareto front when solving nonlinear scalarized MDPs.
\begin{theorem}\label{theorem:nonlinear scalarization}
    Let $f$ be strictly increasing and let $\pi^*$ be an optimal policy for the scalarized MDP. Then $\pi^*$ is Pareto optimal and can be represented as a convex combination of at most $K$ Pareto-optimal deterministic policies.
\end{theorem}
\begin{proof}
    Suppose $\pi^*$ is not Pareto optimal. Then there exists an objective vector $\V{J}_{\beta}(\pi)$ that strictly dominates $\V{J}_{\beta}(\pi^*)$, i.e., $\V{J}_{\beta}(\pi) \preceq \V{J}_{\beta}(\pi^*)$ and $\V{J}_{\beta}(\pi) \neq \V{J}_{\beta}(\pi^*)$. Since $f$ is strictly increasing, we have $f(\V{J}_{\beta}(\pi)) < f(\V{J}_{\beta}(\pi^*))$. This contradicts the optimality of $\pi^*$ for the scalarized problem. Thus, $\pi^*$ is Pareto optimal.

    Recall from Theorem~\ref{theorem:frontier} that the Pareto front lies on the boundary of the $K$-dimensional convex polytope $\mathcal{J}^\Phi$. Each face of $\mathcal{J}^\Phi$ has dimension $d \leq K-1$. By Carath\'eodory's theorem~\cite[Sec.~17]{rockafellar1997convex}, any point in a $d$-dimensional face can be expressed as a convex combination of at most $d+1$ vertices of that face. Since each vertex of $\mathcal{J}^\Phi$ corresponds to a deterministic policy, $\pi^*$ can be represented as a mixture of at most $K$ Pareto-optimal deterministic policies.
\end{proof}


\begin{figure}[t!]
    \centering
    \scalebox{0.74}{\begin{tikzpicture}[scale=1.0]
    \node [draw, rounded corners=1.5pt, rectangle, minimum width=1.5cm, minimum height=0.8cm, thick] (source) {Process};
    \node [draw, rounded corners=1.5pt, rectangle, right=0.9cm of source, minimum width=1.5cm, minimum height=0.8cm, thick] (sensor) {Sensor};
    \node [draw, align=center, rounded corners=2pt, rectangle, minimum width=1.5cm, minimum height=0.8cm, right=2cm of sensor, thick] (channel) {Channel};
    \node [draw, rounded corners=1.5pt, rectangle, right=0.9cm of channel, minimum width=1.5cm, minimum height=0.8cm, thick] (receiver) {Receiver};
    \node [right=0.8cm of receiver, thick] (end) {};
    \node [coordinate, right=0.8 of sensor, thick] (switch) {};
    
    \draw[->, thick] (source) -- node[above] {\large $Z_t$} (sensor);
    \draw[-, thick] (sensor) -- node[above] {\large $\hat{Z}^{\T{loc}}_t$} (switch);
    \draw[-, thick] (switch) -- ++(0.55, 0.5) {};
    \draw[-, thick] (switch)++(0.5, 0) -- node[above] {} (channel);
    \draw[->, bend left, thick] ($(switch)+(0.3,0.5)$) to ($(switch)+(0.7,0.)$);
    \draw[->, thick] (channel) -- node[above] {} coordinate[midway] (mid-cd) (receiver);
    \draw[->, thick] (receiver) -- node[above] {\large $\hat{Z_t}$} (end);
    
    \draw[->, dashed, thick] (mid-cd) -- ++(0,-0.8) -| 
    node[pos=0.25, below] {} 
    ($(sensor.south west)!0.5!(sensor.south east)$);
\end{tikzpicture}}
    \vspace{-0.1in}
    \caption{Remote state estimation of a linear Gaussian process.}
    \label{fig:system_model}
\end{figure}
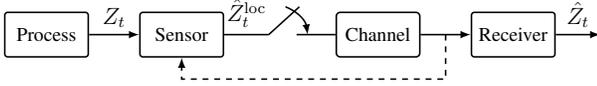

\section{Case Study: Remote Estimation of Linear Gaussian Processes}\label{sec:example}
\subsection{Problem Description}
In this section, we apply our structural results to a remote estimation problem, e.g.,~\cite{schenato2008optimal, leong2017sensor, luo2025remote}. Consider a point-to-point communication system depicted in Figure~\ref{fig:system_model}, where a sensor observes a linear Gaussian process
\begin{equation*}
    Z_{t+1} = \M{A}Z_t + W_t, 
\end{equation*}
where $Z_t \in \mathbb{R}^{n}$ is the system state at time $t$, $\M{A} \in \mathbb{R}^{n \times n}$ is the state transition matrix, and $W_t \in \mathbb{R}^{n}$ is zero-mean Gaussian process noise with covariance $\M{Q} \succeq \V{0}$. The sensor measurements are given by $Y_t = \M{C} Z_t + V_t$, where $\M{C} \in \mathbb{R}^{m \times n}$ is the measurement matrix and $V_t \in \mathbb{R}^m$ is zero-mean Gaussian measurement noise with covariance $\M{R} \succ \V{0}$. The sensor employs a steady-state Kalman filter to produce a local estimate $\hat{Z}_t^{\T{loc}}$ with a \emph{constant} error covariance $\bar{\mathbf{K}}$.

Due to energy constraints, the sensor can transmit only intermittently. Let $U_t \in \{0,1\}$ denote the sensor's action: the sensor sends the local estimate $\hat{Z}^{\T{loc}}_t$ to the receiver when $U_t = 1$ and remains silent when $U_t = 0$. We assume an i.i.d. packet-dropping channel $\{H_t\}$, where each transmitted packet is successfully delivered with probability ${\sf E}[H_t] = p_\T{s}$. The sensor is aware of the channel outcome through a feedback link. The receiver's estimate then evolves as
\begin{equation*}
    \hat{Z}_{t} = U_tH_t \hat{Z}^\T{loc}_{t} + 
		(1 - U_tH_t)\M{A} \hat{Z}_{t-1}.
\end{equation*}
Let $X_t := t - \max\{\tau \leq t : U_\tau = 1\}$ denote the time elapsed since the last successful status update. To avoid technicalities, we impose $X_t = \min\{X_t, X_{\max}\}$. The error covariance at the receiver is then given by
\begin{equation*}
    \M{K}_{t} = U_tH_t \bar{\M{K}} + (1 - U_tH_t) (\M{A}\M{K}_{t-1}\M{A}^\top + \M{Q}).
\end{equation*}
Moreover, we can write $\operatorname{Tr}(\M{K}_t) = \eta(X_t)$, where $\eta: \mathbb{N} \to \mathbb{R}$ is an increasing function parametrized by $\bar{\M{K}}$, $\M{A}$, and $\M{Q}$.

The goal is to find a communication policy $\pi$ that balances the average estimation error $J_{1}(\pi)$ and the average communication rate $J_{2}(\pi)$, where
\begin{align*}
    J_{1}(\pi) &:= \limsup_{T \to \infty}\frac{1}{T}\sum_{t=1}^T {\sf E}^\pi \big[\eta(X_t) \mid X_1 = 0 \big],\\
    J_{2}(\pi) &:= \limsup_{T \to \infty}\frac{1}{T}\sum_{t=1}^T {\sf E}^\pi \big[U_t \mid X_1 = 0 \big].
\end{align*}
A common approach to characterize this trade-off is linear scalarization~\cite{leong2017sensor, luo2025remote}, i.e.,
\begin{equation*}
    \min_{\pi \in \Pi}\, J(\pi; \lambda) = J_1(\pi) + \lambda J_2(\pi),
\end{equation*}
where $\lambda > 0$ denotes the cost per transmission.

\begin{lemma}[{\hspace{-0.03em}\cite{leong2017sensor}}]
\label{lemma:threshold}
    Let $\|\M{A}\|^2(1 - p_\T{s}) < 1$. Then the optimal policy $\pi^*$ for the linear scalarized MDP is a threshold policy; that is, $\pi^*(x) = 1$ if $x \geq x_\T{th}$ and $\pi^*(x) =0$ otherwise.
\end{lemma}

\begin{figure*}[t!]
    \centering
    \begin{minipage}{0.49\textwidth}
        \centering
        \includegraphics[width=0.85\linewidth]{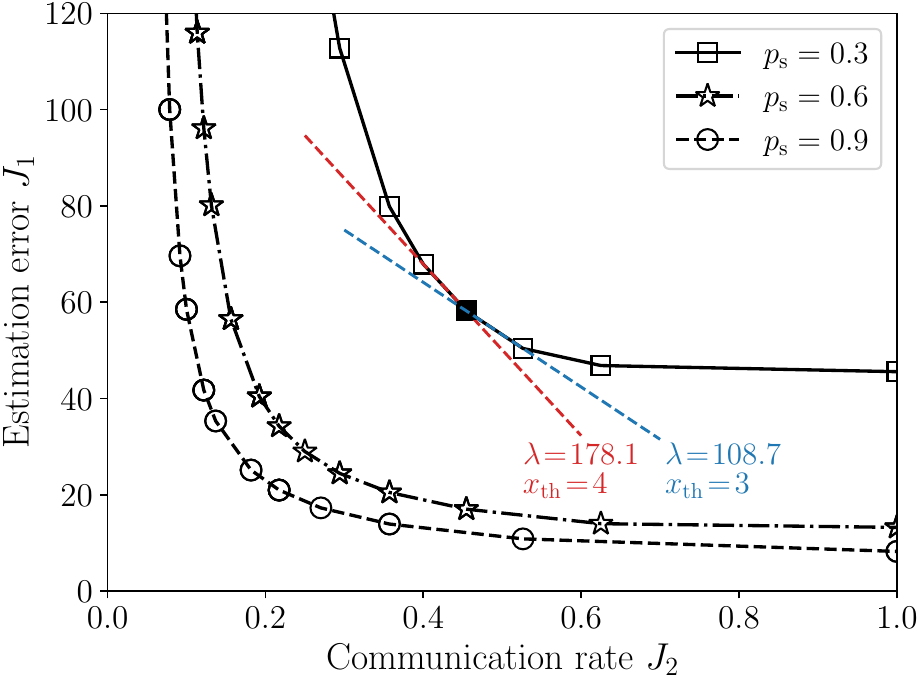}
        \vspace{-0.1in}
        \caption{Pareto front of the estimation system.}
        \label{fig:front}
    \end{minipage}
    \hfill
    \begin{minipage}{0.49\textwidth}
        \centering
        \includegraphics[width=0.85\linewidth]{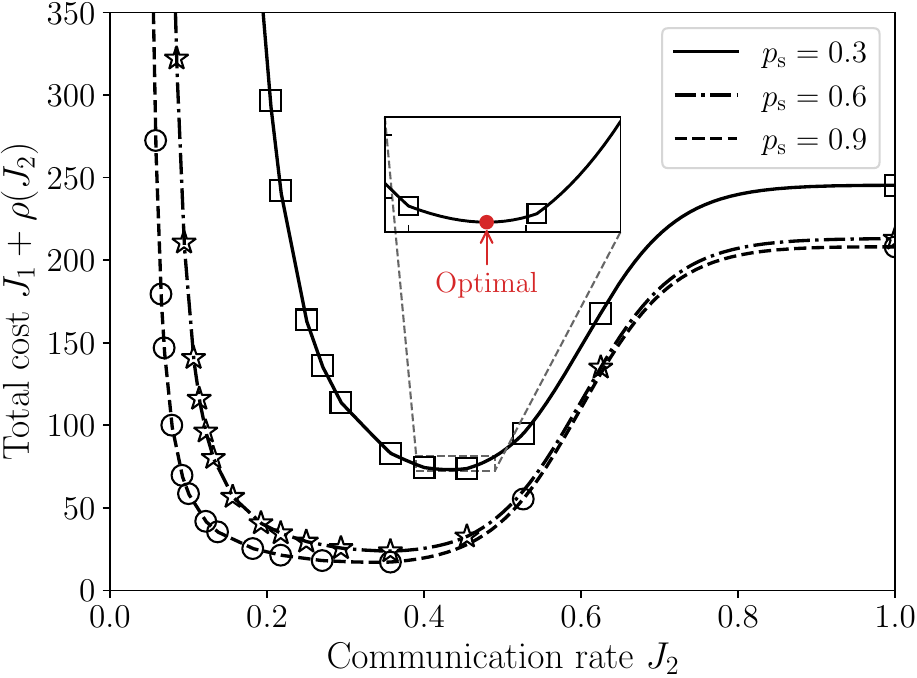}
        \vspace{-0.1in}
        \caption{Achievable total cost for the nonlinear scalarized problem.}
        \label{fig:transformed_front}
    \end{minipage}
\end{figure*}

However, this approach has two key limitations: it yields only a single policy for each preference configuration and fails to capture nonlinear operational costs. For instance, frequent communication may accelerate hardware aging and reduce link reliability. We therefore study two alternative formulations. The first is the MOMDP
\begin{equation*}
    \min_{\pi \in \Pi}\, \V{J}(\pi) = \big(J_{1}(\pi), J_{2}(\pi)\big).
\end{equation*}
The second is a nonlinear scalarization
\begin{equation*}
    \min_{\pi \in \Pi}\, f(\V{J}(\pi)) =  J_{1}(\pi) +  \rho(J_{2}(\pi)),
\end{equation*}
where $\rho: [0, 1] \to \mathbb{R}$ is a strictly increasing function.

\subsection{Analysis and Numerical Results}
From Theorem~\ref{theorem:frontier}, the Pareto front is a piecewise linear, convex curve in the 2D objective plane. Each vertex corresponds to a deterministic policy that solves the linear scalarized MDP for some $\lambda > 0$. Lemma~\ref{lemma:threshold} implies that, although each vertex has $X_{\max}$ adjacent policies, only \emph{two} neighboring threshold policies need to be evaluated. Hence, the Pareto front forms a directed path along the boundary of the polytope and can be computed without solving any MDP. These results are illustrated using the following example. 

Consider the pendubot system in~\cite{schenato2007foundations}, where
\begin{align*}
    \M{A} &= \begin{bmatrix}
    1.0058 & 0.0150 & -0.0016 & 0.0000 \\
    0.7808 & 1.0058 & -0.2105 & -0.0016 \\
    -0.0060 & 0.0000 & 1.0077 & 0.0150 \\
    -0.7962 & -0.0060 & 1.0294 & 1.0077
\end{bmatrix}, \\
\M{C} &= \begin{bmatrix}
1 & 0 & 0 & 0 \\
0 & 0 & 1 & 0
\end{bmatrix}, \M{Q} = \V{q}\V{q}^\top, \M{R} = 0.001\times \M{I},\\
\V{q} &= \begin{bmatrix}
    0.003 & 1 & -0.005 & -2.150
\end{bmatrix}^\top.
\end{align*}

Figure~\ref{fig:front} illustrates the Pareto front of the above system. The absolute slope ($\lambda$) reflects the \emph{value} of communication: its marginal benefit diminishes at high communication rates, while careful calibration is critical when communication is scarce. For $p_\T{s} = 0.3$ and $\lambda \in (108.7, 178.1)$, a single vertex (filled black square) attains the global optimum for all $\lambda$ in this interval. This highlights that linear scalarization is inefficient, as it requires sweeping over a wide range of weights to recover different parts of the front. Moreover, any point on an edge can be obtained by mixing two adjacent threshold policies.

Figure~\ref{fig:transformed_front} shows the total cost on the Pareto front with a sigmoid penalty $\rho(J_2) = 200 / (1 + \exp(-17 (x - 0.6)))$. The optimal policy is found at an edge of the Pareto front.

\appendices

\section{Proof of Lemma~\ref{lemma:singleton}}\label{proof:lemma-singleton}
Let $\pi \in \Pi_{\T S}$ be a stationary policy. Then the state process $\{X_t\}$ governed by $\M{P}_\pi$ forms a time-homogeneous Markov chain. By the Chapman-Kolmogorov equations, 
\begin{equation*}
    {\sf P}^\pi(X_n = x| X_1 = i) = \M{P}_\pi^{n-1}(x|i).
\end{equation*}
We may write the finite-horizon state-action frequency as
\begin{align*}
    \mu_{\beta, \pi}^t(x, u) 
    &= \frac{1}{t}\sum_{n=1}^t\sum_{i\in {\sf X}} \beta(i) {\sf P}^\pi (X_n = x, U_n = u|X_1 = i)\\
    &= \frac{1}{t}\sum_{n=1}^t\sum_{i\in {\sf X}} \beta(i)\pi(u|x)\M{P}_\pi^{n-1}(x|i).
\end{align*}
For any finite Markov chain, the Ces\`{a}ro limit
\begin{equation*}
    \lim_{t\to \infty} \frac{1}{t}\sum_{n=1}^t \M{P}_\pi^{n-1}(x|i) = \M{P}^*_\pi(x|i)
\end{equation*}
always exists. Thus, the occupancy measure converges to
\begin{equation}
    \mu_{\beta, \pi}(x, u) = \sum_{i\in {\sf X}}\beta(i)
    \pi(u|x)\M{P}_\pi^{*}(x|i). \label{eq:cesaro limit}
\end{equation}
Since the ordinary limit exists, $\Gamma_\beta^\pi = \{\mu_{\beta, \pi}\}$ is a singleton. 

Under the unichain assumption, we have $\M{P}^*_\pi(x|i) = \nu_\pi(x)$ for all $i \in {\sf X}$, where $\nu_\pi$ is the unique stationary distribution. Substituting this into~\eqref{eq:cesaro limit} gives
\begin{equation}
    \mu_{\beta, \pi}(x, u) 
    = \pi(u|x)\nu_{\pi}(x), 
\end{equation}
which is independent of $\beta$. This establishes the result.

\section{Proof of Theorem~\ref{theorem:occupancy measure}}\label{proof:theorem-occupancy-measure}
We classify the constraints in~\eqref{eq:Phi set} as follows:
\begin{itemize}
    \item A set of $|{\sf X}|$ \textit{balance} constraints
    \begin{equation*}
        \sum_{u}\phi(x^\prime, u) = \sum_{x, u}\phi(x, u)\M{P}(x, u, x^\prime), \,\, x^\prime \in {\sf X}.
    \end{equation*}
    \item A set of $|{\sf X}||{\sf U}|$  \emph{nonnegativity} constraints
    \begin{equation*}
        \phi(x, u) \geq 0, \,\, (x, u)\in {\sf X}\times {\sf U}.
    \end{equation*}
    \item One \emph{normalization} constraint $\sum_{x, u}\phi(x, u) = 1$.
\end{itemize}

Moreover, there are $|{\sf X}| - 1$ linearly independent balance constraints, which are also independent of the normalization constraint. Thus, the rank of the equality constraints is $|{\sf X}|$. Since a vertex requires $|{\sf X}||{\sf U}|$ active linearly independent constraints~\cite[Theorem~2.3]{bertsimas1997introduction}, at least $|{\sf X}|(|{\sf U}| - 1)$ nonnegativity constraints must be active at any vertex of $\Phi$. Moreover, two vertices are adjacent if and only if they share $|{\sf X}||{\sf U}|-1$ linearly independent active constraints.

\subsection{Proof of Theorem~\ref{theorem:occupancy measure}(a)}
Theorem~\ref{theorem:occupancy measure}(a) establishes the equivalence between the set of occupancy measures $\Gamma_\beta$ and the auxiliary set $\Phi$. The following lemma is a key step in the proof, which holds without the unichain assumption.

\begin{lemma}\label{lemma:multi-chain polytope}
    $\Gamma_\beta = \operatorname{conv}(\Gamma_\beta^{\Pi_{\T S}}) = \operatorname{conv}(\Gamma_\beta^{\Pi_{\T D}})$.
\end{lemma}
\begin{proof}
Clearly, for any $\beta$, $\Gamma_\beta^{\Pi_{\T D}} \subset \Gamma_\beta^{\Pi_{\T S}} \subset \Gamma_\beta$. The proof is then divided into three parts.

\textit{Part i):} $\Gamma_\beta \supset \operatorname{conv}(\Gamma_\beta^{\Pi_{\T D}})$. Let $\Pi_{\T D} = \{\pi_1, \ldots, \pi_m\}$. For some $\V{\alpha} \succeq 0$ with $\sum_{i=1}^m\alpha_i = 1$, let $\hat{\pi} \in \Pi_{\T M}$ be an initial mixing policy which selects $\alpha_i$ with probability $\alpha_i$ at $t=1$ and follows it thereafter. Then, for all $n \geq 1$,
\begin{equation*}
    {\sf P}_\beta^{\hat{\pi}}(X_n = x, U_n = u) = \sum_{i=1}^m \alpha_i {\sf P}_\beta^{\pi_i}(X_n = x, U_n = u).
\end{equation*}
The finite-horizon state-action frequency $\mu_{\beta, \hat{\pi}}^t$ satisfies
\begin{align*}
    \mu_{\beta, \hat{\pi}}^t(x, u) &= \frac{1}{t}\sum_{n=1}^t\sum_{i=1}^m \alpha_i {\sf P}_\beta^{\pi_i}(X_n = x, U_n = u)\\
    &= \sum_{i=1}^m \alpha_i \mu_{\beta, \pi_i}^t(x, u).
\end{align*}
Since each $\pi_i$ is stationary, the sequence $\{\mu_{\beta, \pi_i}^t(x, u)\}$ converges to $\mu_{\beta, \pi_i}(x, u)$, which gives that
\begin{equation*}
    \mu_{\beta, \hat{\pi}}(x, u) = \sum_{i=1}^m \alpha_i \mu_{\beta, \pi_i}(x, u).
\end{equation*}
Hence, $\operatorname{conv}(\Gamma_\beta^{\Pi_{\T D}}) \subset \Gamma_\beta$. This proves Part (i).

\textit{Part ii):} $\Gamma_\beta \subset \operatorname{conv}(\Gamma_\beta^{\Pi_{\T D}})$. Suppose, to the contrary, that there exists $\mu^\prime \in \Gamma_\beta$ such that $\mu^\prime \notin \operatorname{conv}(\Gamma_\beta^{\Pi_{\T D}})$. Since $\Gamma_\beta^{\Pi_{\T D}} = \{\mu_\pi: \pi \in \Pi_{\T D}\}$ is finite, its convex hull is compact. By the hyperplane separation theorem~\cite[Sec.~11]{rockafellar1997convex}, there exists a hyperplane that separates $\{\mu^\prime\}$ from $\operatorname{conv}(\Gamma_\beta^{\Pi_{\T D}})$; that is, there exists a nonzero vector $\tilde{c} \in \mathbb{R}^{|{\sf X}||{\sf U}|}$ such that
\begin{equation}
    \langle \mu^\prime, \tilde{c} \rangle < \inf_{\mu \in \operatorname{conv}(\Gamma_\beta^{\Pi_{\T D}})} \langle \mu, \tilde{c} \rangle. \label{eq:separation-1}
\end{equation}
Consider an MDP described by $({\sf X}, {\sf U}, \M{P}, \tilde{c}, \beta)$. It is known (e.g., Puterman~\cite[Ch.~8]{puterman1994markov}) that there exists an optimal deterministic policy $\pi^*\in \Pi_{\T D}$ for this MDP. This implies that
\begin{equation}
    \langle \mu_{\pi^*}, \tilde{c} \rangle 
    = \inf_{\mu \in \Gamma_\beta^{\Pi_{\T D}}}\, \langle \mu, \tilde{c} \rangle 
    \leq \inf_{\mu \in \Gamma_\beta} \,\langle \mu, \tilde{c} \rangle 
    \leq \langle \mu^\prime, \tilde{c} \rangle, \label{eq:separation-2}
\end{equation}
which contradicts~\eqref{eq:separation-1}. This establishes Part (ii).

\textit{Part iii):} $\operatorname{conv}(\Gamma_\beta^{\Pi_{\T S}}) = \operatorname{conv}(\Gamma_\beta^{\Pi_{\T D}})$. Since $\Gamma_\beta^{\Pi_{\T D}} \subset \Gamma_\beta^{\Pi_{\T S}}$, we have $\operatorname{conv}(\Gamma_\beta^{\Pi_{\T D}}) \subset \operatorname{conv}(\Gamma_\beta^{\Pi_{\T S}})$. Conversely, from Part (ii), $\Gamma_\beta^{\Pi_{\T S}} \subset \Gamma_\beta \subset \operatorname{conv}(\Gamma_\beta^{\Pi_{\T D}})$. Taking convex hulls yields
\begin{equation*}
    \operatorname{conv}(\Gamma_\beta^{\Pi_{\T S}}) \subset \operatorname{conv}\big(\operatorname{conv}(\Gamma_\beta^{\Pi_{\T D}})\big) = \operatorname{conv}(\Gamma_\beta^{\Pi_{\T D}}).
\end{equation*}
This proves the lemma.
\end{proof}

Returning to Theorem~\ref{theorem:occupancy measure}(a), it remains to show that $\Gamma_\beta^{\Pi_{\T S}} = \Phi$. This equivalence, combined with the fact that $\Phi$ is a convex polytope, establishes this assertion. This part is formalized in the following lemma.

\begin{lemma}\label{lemma:unichain polytope}
    For all $\beta$, $\Gamma_\beta^{\Pi_{\T S}} = \Phi$. Each $\mu\in \Gamma_\beta^{\Pi}$ is realizable by a stationary policy defined by~\eqref{eq:recover policy}.
\end{lemma}
\begin{proof}
We establish the equivalence in two steps. 

\textit{Part i):} $\Gamma_\beta^{\Pi_{\T S}} \subset \Phi$. That is, for each $\mu_{\pi} \in \Gamma_\beta^{\Pi_{\T S}}$ induced by some stationary policy $\pi \in \Pi_{\T S}$, we show that $\mu_{\pi} \in \Phi$. 

The finite-horizon state-action frequency $\mu_{\beta, \pi}^t$ satisfies the normalization and non-negativity conditions
\begin{equation}
    \sum_{x, u}\mu_{\beta, \pi}^t(x, u) = 1\,\,\text{and}\,\, \mu_{\beta, \pi}^t(x, u) \geq 0 \label{eq:normalization condition}
\end{equation}
for all $t\geq 1$. For any $x\in {\sf X}$, the state frequency satisfies 
\begin{align}
    &\mu_{\beta, \pi}^t(x^\prime)
    = \frac{1}{t}\sum_{n=1}^t{\sf P}_\beta^\pi(X_n = x^\prime) \notag\\
    =~& \frac{1}{t}\Big[\sum_{n=2}^{t+1}\sum_{x, u}{\sf P}_\beta^\pi(X_{n-1}=x, U_{n-1}=u) \M{P}(x, u, x^\prime)\notag\\
    & + \beta(x^\prime) - \sum_{x, u}{\sf P}_\beta^\pi(X_{t}=x, U_{t}=u) \M{P}(x, u, x^\prime)
    \Big]\notag\\
    =~&\frac{1}{t}\sum_{n^\prime =1}^{t}\sum_{x, u}{\sf P}_\beta^\pi(X_{n^\prime}=x, U_{n^\prime}=u) \M{P}(x, u, x^\prime)\notag\\
    & + \frac{\beta(x^\prime)}{t} - \frac{\sum_{x, u}{\sf P}_\beta^\pi(X_{t}=x, U_{t}=u) \M{P}(x, u, x^\prime)}{t} \notag\\
    =~& \sum_{x, u} \mu_{\beta, \pi}^t(x, u) \M{P}(x, u, x^\prime)\notag\\
    & + \frac{\beta(x^\prime)}{t} - \frac{\sum_{x, u}{\sf P}_\beta^\pi(X_{t}=x, U_{t}=u) \M{P}(x, u, x^\prime)}{t}.\label{eq:balance condition}
\end{align}
From Lemma~\ref{lemma:singleton}, the limit $\mu_{\pi}(x, u) = \lim_{t\to \infty}\mu_{\beta, \pi}^t(x, u)$ exists. Taking limits at both sides of~\eqref{eq:balance condition} gives
\begin{equation}
    \sum_{u}\mu_{\pi}(x^\prime, u) = \sum_{x, u} \mu_{\pi}(x, u)\M{P}(x, u, x^\prime).
\end{equation}
This, together with~\eqref{eq:normalization condition}, implies $\mu_{\pi} \in \Phi$. Hence, $\Gamma_\beta^{\Pi_{\T S}} \subset \Phi$.

\textit{Part ii):} $\Gamma_\beta^{\Pi_{\T S}} \supset \Phi$. That is, for each $\phi \in \Phi$, there exists some stationary policy $\pi_\phi \in \Pi_{\T S}$ such that $\mu_{\pi_{\phi}} \in \Gamma_\beta^{\Pi_{\T S}}$. 

For any $\phi \in \Phi$, define the stationary policy $\pi_\phi \in \Pi_{\T S}$ as
\begin{equation*}
    \pi_{\phi}(u|x) = \begin{cases}
        \phi(x, u) / \phi(x), &\phi(x)>0,\\
        {\rm arbitrary}, &\phi(x)=0.
    \end{cases} 
\end{equation*}
From the definition of $\Phi$, we have
\begin{align}
    \phi(x^\prime) &= \sum_{x, u}\phi(x, u)\M{P}(x, u, x^\prime)\notag\\
    &= \sum_{x}\phi(x)\sum_{u}\pi_{\phi}(u|x)\M{P}(x, u, x^\prime)\notag\\
    &= \sum_{x}\phi(x)\M{P}_{\pi_\phi}(x,x^\prime).\label{eq:balance equation}
\end{align}
Under the unichain assumption, $\{X_t\}$ under $\pi_\phi$ has a unique stationary distribution $\nu_{\pi_\phi}$ that satisfies~\eqref{eq:balance equation}. Since $\phi(x)$ satisfies~\eqref{eq:balance equation} and $\sum_x\phi(x) = 1$, this implies that $\phi(x) = \nu_{\pi_\phi}(x)$ for all $x$. It follows from Lemma~\ref{lemma:singleton} that 
\begin{equation*}
    \mu_{\pi_\phi}(x, u) = \nu_{\pi_\phi}(x)\pi_{\phi}(u|x) = \phi(x, u).
\end{equation*}
This confirms that $\phi$ is realizable by $\pi_\phi$, hence $\Phi \subset \Gamma_\beta^{\Pi_{\T S}}$. 
\end{proof}

\subsection{Proof of Theorem~\ref{theorem:occupancy measure}(b)}\label{proof:theorem-occupancy-measure-b}
Theorem~\ref{theorem:occupancy measure}(b) states that: (i) any deterministic policy produces an occupancy measure that is a vertex of $\Phi$; and (ii) any vertex of $\Phi$ corresponds to a deterministic policy.

\textit{Part i):} $\Gamma_\beta^{\Pi_{\T D}} \subset \mathcal{V}(\Phi)$. Let $\pi\in \Pi_{\T D}$. Suppose that $\mu_\pi \notin \mathcal{V}(\Phi)$. Then there exist two distinct vectors $\mu_1, \mu_2 \in \Phi$ and a scalar $\alpha \in (0, 1)$ such that $\mu_\pi = \alpha \mu_1 + (1-\alpha) \mu_2$. Since $\pi \in \Pi_{\T D}$, its occupancy measure satisfies
\begin{equation*}
    \mu_\pi(x, u) = \begin{cases}
        \nu_{\pi}(x), &u = \pi(x),\\
        0, &u \neq \pi(x).
    \end{cases}
\end{equation*}
Hence, for all $x \in {\sf X}$ and all $u\neq \pi(x)$, we have
\begin{equation*}
    \alpha \mu_1(x, u) + (1-\alpha) \mu_2(x, u) = 0.
\end{equation*}
Because $\mu_i(x, u)\geq 0$, it follows that $\mu_1(x, u) = \mu_2(x, u) = 0$ for all $x \in {\sf X}$ and $u\neq \pi(x)$. This activates $|{\sf X}|(|{\sf U}| - 1)$ nonnegativity constraints at $\mu_\pi$, and these constraints are linearly independent. In addition, there are $|{\sf X}|$ linearly independent equality constraints active at $\mu_\pi$. Therefore, $\mu_\pi$ is a vertex of $\Phi$, which contradicts the assumption. Hence, we conclude that $\Gamma_\beta^{\Pi_{\T D}} \subset \mathcal{V}(\Phi)$.

\textit{Part ii):} $\mathcal{V}(\Phi) \subset \Gamma_\beta^{\Pi_{\T D}}$. Suppose, to the contrary, that there exists a randomized policy $\pi \notin \Pi_{\T D}$ such that $\mu_\pi \in \mathcal{V}(\Phi)$. Then there exists at least one state $i_0 \in {\sf X}$ for which $\pi(\cdot|i_0)$ is not a Dirac measure. Consequently, there exist two distinct distributions $\psi_1, \psi_2$ and a scalar $\alpha \in (0, 1)$ such that
\begin{equation*}
    \pi(u|i_0) = \alpha \psi_1(u) + (1-\alpha) \psi_2(u).
\end{equation*}

Define two stationary policies $\pi_i$, $i=1, 2$, by
\begin{equation}
    \pi_i(u|x) = \begin{cases}
        \pi(u|x), &x \neq i_0,\\
        \psi_i(u), &x = i_0.
    \end{cases} \label{eq:mixing policy}
\end{equation}
That is, $\pi_1$ and $\pi_2$ differ only in state $i_0$. Then the mixing policy $\pi = \alpha \pi_1 + (1-\alpha)\pi_2$ behaves as follows: whenever the system returns to state $i_0$, it selects $\pi_1$ with probability $\alpha$ and $\pi_2$ with probability $1-\alpha$, and follows the selected policy until the next visit to $i_0$.

Similar to the simple mixing policy in Definition~\ref{def:mixing policy}, $\{X_t\}$ induced by $\pi$ forms a regeneration process. Let $T_\T{c}$ denote the regeneration cycle length (i.e., the time between successive visits to $i_0$), and let $N_\T{c}(x, u)$ denote the number of times action $u$ is taken in state $x$ during one regeneration cycle. Because the choice between $\pi_1$ and $\pi_2$ is made only at regeneration epochs, conditioning on this choice yields
\begin{align}
    \mathbb{E}^{\pi}[T_\T{c}]
    &= \alpha \mathbb{E}^{\pi_1}[T_\T{c}] + (1-\alpha)\mathbb{E}^{\pi_2}[T_\T{c}], \notag \\
    \mathbb{E}^{\pi}[N_\T{c}(x, u)]
    &= \alpha \mathbb{E}^{\pi_1}[N_\T{c}(x, u)] + (1-\alpha)\mathbb{E}^{\pi_2}[N_\T{c}(x, u)]. \notag
\end{align}

Since $\{X_t\}$ induced by $\pi$ and $\pi_i$, $i=1, 2$, contain a single recurrent class, their stationary distributions satisfy
\begin{equation*}
    \mu_\pi(x, u) = \frac{\mathbb{E}^\pi[N_\T{c}(x, u)]}{\mathbb{E}^{\pi}[T_\T{c}]},\, \mu_{\pi_i}(x, u) = \frac{\mathbb{E}^{\pi_i}[N_\T{c}(x, u)]}{\mathbb{E}^{\pi_i}[T_\T{c}]}.
\end{equation*}
In particular, $\nu_\pi(i_0) = 1/\mathbb{E}^{\pi}[T_\T{c}]$ and $\nu_{\pi_i}(i_0) = 1/\mathbb{E}^{\pi_i}[T_\T{c}]$.

Then we may write
\begin{align}
    \mu_\pi(x, u)
    &= \frac{\alpha \mathbb{E}^{\pi_1}[N_\T{c}(x, u)] + (1-\alpha)\mathbb{E}^{\pi_2}[N_\T{c}(x, u)]}{\alpha \mathbb{E}^{\pi_1}[T_\T{c}] + (1-\alpha)\mathbb{E}^{\pi_2}[T_\T{c}]}\notag\\
    &=
    \frac{\alpha \mu_{\pi_1}(x, u)
    \mathbb{E}^{\pi_1}[T_\T{c}]
    + (1-\alpha)\mu_{\pi_2}(x, u)\mathbb{E}^{\pi_2}[T_\T{c}]}
    {\alpha \mathbb{E}^{\pi_1}[T_\T{c}] + (1-\alpha)\mathbb{E}^{\pi_2}[T_\T{c}]}\notag\\
    &= b \mu_{\pi_1}(x, u) + (1-b)\mu_{\pi_2}(x, u),
\end{align}
where 
\begin{equation}
    b = 
    \frac{\alpha \nu_{\pi_2}(i_0)}
    {\alpha \nu_{\pi_2}(i_0) + (1-\alpha)\nu_{\pi_1}(i_0)}. \label{eq:b}
\end{equation}
Thus, $\mu_\pi$ is a convex combination of distinct points in $\Phi$ and cannot be a vertex. This establishes the result.

\subsection{Proof of Theorem~\ref{theorem:occupancy measure}(c)}
Let $\mu_{1}, \mu_{2}$ be two distinct vertices, and let $\pi_1, \pi_2 \in \Pi_\T{D}$ denote their corresponding deterministic policies. Two vertices are adjacent if and only if they share $|{\sf X}||{\sf U}| - 1$ independent active constraints. Under any deterministic policy, each state with positive occupancy contributes $|{\sf U}|-1$ linearly independent active nonnegativity constraints, while a transient state (with zero occupancy) contributes $|{\sf U}|$ active constraints, of which $|{\sf U}|-1$ are linearly independent.

\textit{Sufficiency:} Suppose that $\mu_1$ and $\mu_2$ differ only in state $i_0$. Since $\mu_1 \neq \mu_2$, the actions prescribed at $i_0$ are different, and at least one of the corresponding policies assigns positive occupancy to this state. Both vertices satisfy the same $|{\sf X}|$ equality constraints. For each $x\neq i_0$, the policies are identical and thus share $|{\sf U}| - 1$ active nonnegativity constraints. At $x = i_0$, they share $|{\sf U}| - 2$ active nonnegativity constraints. The number of shared independent active constraints is
\begin{equation*}
   |{\sf X}| + (|{\sf X}| - 1)(|{\sf U}| - 1) + |{\sf U}| - 2 = |{\sf X}||{\sf U}| - 1.
\end{equation*}
Hence, $\mu_1$ and $\mu_2$ are adjacent.

\textit{Necessity:} Suppose $\mu_1$ and $\mu_2$ differ in $\ell > 1$ states where at least one of the policies assigns positive occupancy. At each of these $\ell$ states, they share $|{\sf U}|-2$ active nonnegativity constraints. Therefore, the total number of shared linearly independent active constraints is
\begin{equation*}
|{\sf X}| + (|{\sf X}| - \ell)(|{\sf U}| - 1) + \ell(|{\sf U}| - 2) = |{\sf X}||{\sf U}| - \ell.
\end{equation*}
Since $\ell > 1$, this number is strictly less than $|{\sf X}||{\sf U}| - 1$. Hence, $\mu_1$ and $\mu_2$ cannot be adjacent.

\subsection{Proof of Theorem~\ref{theorem:occupancy measure}(d)}
Let $\phi_1, \phi_2 \in \mathcal{V}(\Phi)$ be adjacent vertices differing in state $i_0$, and let $\pi_{\phi_1}, \pi_{\phi_2} \in \Pi_{\T D}$ denote the corresponding deterministic policies. Define a simple mixing policy
\begin{equation*}
    \pi = \alpha \pi_{\phi_1} + (1-\alpha) \pi_{\phi_2}, 
\end{equation*}
which selects $\pi_{\phi_1}$ with probability $\alpha$ and $\pi_{\phi_2}$ with probability $1-\alpha$ whenever state $i_0$ is visited.

Following the regeneration argument in Appendix~\ref{proof:theorem-occupancy-measure-b}, the occupancy measure induced by $\pi$ satisfies
\begin{equation*}
    \mu_\pi = b \phi_1 + (1-b) \phi_2 \in \mathcal{G}(\phi_1, \phi_2), 
\end{equation*}
where $b$ is given by~\eqref{eq:b}. Since the mapping $\alpha \mapsto b$ is continuous and strictly monotone, every point on the edge $\mathcal{G}(\phi_1, \phi_2)$ can be realized by a suitable choice of $\alpha$. This establishes the result.

\section{Proof of Theorem~\ref{theorem:frontier}}\label{proof:theorem-frontier}
\textit{Part (a):} Let $\mathcal{V}(\Phi) = \{v_1, \ldots, v_m\}$. Since $\Phi$ is a convex polytope, any point $\phi \in \Phi$ can be written as 
\begin{equation*}
    \phi = \sum_{i=1}^m \lambda_i v_i,\,\,\, \text{where}~\lambda_i \geq 0~\text{and}~\sum_{i=1}^m \lambda_i = 1.
\end{equation*}
By the linearity of the operator $\V{J}$, we have
\begin{equation}
    \V{J}(\phi) = \V{J}\big(\sum_{i=1}^m \lambda_i v_i\big) = \sum_{i=1}^m \lambda_i \V{J}(v_i). \label{eq:prove-linearity}
\end{equation}
This shows that every objective vector $\V{J}(\phi) \in \mathcal{J}^\Phi$ is a convex combination of the points in $\mathcal{J}^{\mathcal{V}(\Phi)} = \{\V{J}(v_i) : v_i \in \mathcal{V}(\Phi)\}$. Hence, $\mathcal{J}^\Phi = \operatorname{conv}(\mathcal{J}^{\mathcal{V}(\Phi)})$ is a convex polytope, and its vertex set satisfies $\mathcal{V}(\mathcal{J}^\Phi) \subseteq \mathcal{J}^{\mathcal{V}(\Phi)}$. 

\textit{Part (b):} From Theorem~\ref{theorem:occupancy measure}(a), we have $\mathcal{J}_\beta = \mathcal{J}_\beta^{\Pi_{\T S}} = \mathcal{J}^\Phi$. It follows from~\eqref{eq:prove-linearity} that
\begin{align}
    \operatorname{conv}(\mathcal{J}_\beta^{\Pi_{\T D}})
    &= \operatorname{conv}(\{\V{J} (\phi): \phi \in \Gamma_\beta^{\Pi_{\T D}}\}) \notag\\
    &= \operatorname{conv}(\{\V{J} (\phi): \phi \in \mathcal{V}(\Phi)\}) \notag\\
    &= \{\V{J}(\phi): \phi \in \operatorname{conv}(\mathcal{V}(\Phi))\} \notag\\
    &= \{\V{J}(\phi): \phi \in \Phi\} = \mathcal{J}^{\Phi},
\end{align}
which yields the result.

\textit{Part (c):} Let $\phi^* \in \Phi$ be a Pareto optimal vertex. Suppose that $\V{J}(\phi^*)$ lies in the interior of $\mathcal{J}^\Phi$. Then there exists an open ball $B_\sigma(\V{J}(\phi^*))$ of radius $\sigma>0$ such that $B_\sigma(\V{J}(\phi^*)) \subset \mathcal{J}^\Phi$. We can then choose a unit vector $\V{d} = \frac{1}{\sqrt{K}}[1, \ldots, 1]^\top$ and a sufficiently small scalar $\epsilon>0$ such that the point $\V{y} = \V{J}(\phi^*) - \epsilon \V{d}$ lies within $B_\sigma(\V{J}(\phi^*))$. Since $B_\sigma(\V{J}(\phi^*)) \subset \mathcal{J}^\Phi$, there exists some $\phi \in \Phi$ such that $\V{J}(\phi) = \V{y}$. It follows that $\V{J}(\phi) \prec \V{J}(\phi^*)$, which contradicts the Pareto optimality of $\phi^*$. Thus, $\V{J}(\phi^*)$ must lie on the boundary of $\mathcal{J}^\Phi$.

\textit{Part (d):} Since $\Phi$ is a convex polytope and $\V{J}$ is a linear mapping, there exists a vertex $\phi^* \in \mathcal{V}(\Phi)$ and a nonzero vector $\V{w} \succeq \V{0}$ such that 
\begin{equation}
    \phi^* \in \argmin_{\phi \in \Phi} \,\langle \V{w}, \V{J}(\phi) \rangle.\label{eq:single weight}
\end{equation}
The set of minimizers is given by
\begin{equation}
    \Phi^* := \{\phi \in \Phi: \langle \V{w}, \V{J}(\phi) \rangle = g^*\},
\end{equation}
where $g^* = \langle \V{w}, \V{J}(\phi^*) \rangle$. Clearly, $\Phi^*$ is a face of $\Phi$.

We next show that if $\V{w}\succ \V{0}$, then $\phi^*$ is Pareto optimal. Suppose that $\phi^*$ is not Pareto optimal. Then there exists some $\phi^\prime \in \Phi$ such that $\V{J}(\phi^\prime) \preceq \V{J}(\phi^*)$ and $\V{J}(\phi^\prime) \neq \V{J}(\phi^*)$. Since $\V{w}$ is strictly positive in all components, we obtain
\begin{equation}
    J(\phi^\prime; \V{w}) = \langle \V{w}, \V{J}(\phi^\prime) \rangle < 
    \langle \V{w}, \V{J}(\phi^*) \rangle =  g^*,
\end{equation}
contradicting~\eqref{eq:single weight}. This implies that $\phi^*$ is Pareto optimal. 

\textit{Part (e):} Let $\phi^* \in \mathcal{V}^*(\Phi)$ be a Pareto optimal vertex. From Part (d), $\phi^*$ is a minimizer of the linear functional $J(\phi; \mathbf{w})$ for some $\mathbf{w} \succ \mathbf{0}$. The set of all such minimizers, $\Phi^*$, is a $d$-dimensional Pareto optimal face of $\Phi$.

We first argue that there exists some $\mathbf{w} \succ \mathbf{0}$ such that $d \geq 1$. Suppose, to the contrary, that $\Phi^* = \{\phi^*\}$ for all $\mathbf{w} \succ \mathbf{0}$. This would imply that $\phi^*$ is the unique global optimum for every individual objective $J_k$, which contradicts the assumption that the objectives are in conflict.

Since $\Phi^*$ has dimension $d \geq 1$, the vertex $\phi^* \in \Phi^*$ must be incident to at least one edge $\mathcal{G}(\phi^*, \phi^\prime) \subseteq \Phi^*$. It follows from Part (d) that $\phi^\prime \in \mathcal{N}^*(\phi^*)$, and hence $|\mathcal{N}^*(\phi^*)| \geq 1$.

The upper bound $|\mathcal{N}^*(\phi^*)| \leq |{\sf X}|(|\sf U| - 1)$ follows because each vertex of $\Phi$ corresponds to a deterministic policy, and there are at most $|{\sf X}|(|\sf U| - 1)$ adjacent deterministic policies.

\textit{Part (f):} Let $\phi, \phi^\prime \in \mathcal{V}^*(\Phi)$ be two adjacent Pareto optimal vertices. From Part (e), there exists a $\mathbf{w} \succ \mathbf{0}$ such that $\mathcal{G}(\phi, \phi^\prime)$ is contained in the Pareto optimal face $\Phi^*$.

\bibliographystyle{IEEEtran}
\bibliography{ref}

\end{document}